\documentclass[a4paper,UKenglish,cleveref,autoref,thm-restate,authorcolumns]{lipics-v2021}
\usepackage[htt]{hyphenat}

\title{Prefix-free parsing for building large tunnelled Wheeler graphs}

\titlerunning{PFP for building large tunnelled WGs} 

\author{Adri\'an Goga\footnote{Corresponding author}}{Department of Computer Science, Faculty of Mathematics, Physics and Informatics, Comenius University, Bratislava, Slovakia}{adrian.goga@fmph.uniba.sk}{}{VEGA grant 1/0463/20; EU Horizon 2020 grant No. 956229 (ALPACA); Comenius University grant for doctoral students No. 422}

\author{Andrej Bal\'a\v{z}}{Department of Applied Informatics, Faculty of Mathematics, Physics and Informatics, Comenius University, Bratislava, Slovakia}{andrej.balaz@fmph.uniba.sk}{}{VEGA grant 1/0538/22; EU Horizon 2020 grant No. 956229 (ALPACA)}

\authorrunning{A.\,Goga and A.\, Bal\'a\v{z}} 

\Copyright{Adri\'an Goga and Andrej Bal\'a\v{z}} 



\category{} 

\relatedversion{} 

\supplement{source code: \url{https://github.com/fmfi-compbio/pfp_wg}}

\acknowledgements{We want to thank Travis Gagie for the conception of the idea during his data structures course and helpful remarks throughout the realisation of this project. Our thanks also go to Bro\v{n}a Brejov\'a for useful advice during the writing process and Uwe Baier for kindly responding to our questions via email. Finally, we thank Lucas Pansani Ramos for the aid he provided in the early days of the project.}

\nolinenumbers 

\hideLIPIcs  

\EventEditors{Christina Boucher and Sven Rahmann}
\EventNoEds{2}
\EventLongTitle{22nd International Workshop on Algorithms in Bioinformatics (WABI 2022)}
\EventShortTitle{WABI 2022}
\EventAcronym{WABI}
\EventYear{2022}
\EventDate{September 5--7, 2022}
\EventLocation{Potsdam, Germany}
\EventLogo{}
\SeriesVolume{242}
\ArticleNo{18}

\begin{document}

\maketitle
\begin{abstract}
We propose a new technique for creating a space-efficient index for large repetitive text collections, such as pangenomic databases containing sequences of many individuals from the same species. We combine two recent techniques from this area: Wheeler graphs (Gagie et al., 2017) and prefix-free parsing (PFP, Boucher et al., 2019).

Wheeler graphs are a general framework encompassing several indexes based on the Burrows-Wheeler transform (BWT), such as the FM-index. Wheeler graphs admit a succinct representation which can be further compacted by employing the idea of tunnelling, which exploits redundancies in the form of parallel, equally-labelled paths called blocks that can be merged into a single path. The problem of finding the optimal set of blocks for tunnelling, i.e. the one that minimizes the size of the resulting Wheeler graph, is known to be NP-complete and remains the most computationally challenging part of the tunnelling process. 

To find an adequate set of blocks in less time, we propose a new method based on the prefix-free parsing (PFP). The idea of PFP is to divide the input text into phrases of roughly equal sizes that overlap by a fixed number of characters. The phrases are then sorted lexicographically. The original text is represented by a sequence of phrase ranks (the parse) and a list of all used phrases (the dictionary). In repetitive texts, the PFP representation of the text is generally much shorter than the original since individual phrases are used many times in the parse, thus reducing the size of the dictionary. 

To speed up the block selection for tunnelling, we apply the PFP to obtain the parse and the dictionary of the original text, tunnel the Wheeler graph of the parse using existing heuristics and subsequently use this tunnelled parse to construct a compact Wheeler graph of the original text. Compared with constructing a Wheeler graph from the original text without PFP, our method is much faster and uses less memory on collections of pangenomic sequences. Therefore, our method enables the use of Wheeler graphs as a pangenomic reference for real-world pangenomic datasets.
\end{abstract}

\section{Introduction}
The discovery of Burrows-Wheeler transformation (BWT) \cite{1994burrows&wheeler}, a text permutation initially intended for data compression, and the following breakthrough that enhanced it with indexing properties started a new paradigm of text indexing. In this paradigm, the input text can be stored in space not far larger than its empirical entropy while allowing for efficient pattern matching queries.
The essence of the BWT, the idea of suffix sorting, has been since enjoyed by a large variety of popular genomic indexing tools. Sequencing read alignment tools, such as BWA \cite{2013li}, Bowtie \cite{2010langmead,2012langmead&salzberg} or CHIC \cite{2017valenzuela&makinen} are directly based upon the BWT, while others such as VG \cite{2018garrison_etal} use some of its aspects.

As the BWT is an essential tool in processing large textual datasets, it motivates us to further minimize its space requirements. Since its conception, a popular choice has been to use move-to-front (MTF) transformation and subsequently the run-length encoding (RLE), followed by an entropy encoding, e.g. Huffman \cite{1994burrows&wheeler}. Quite recently, in 2018, Baier \cite{2018baier} 
described another novel source of redundancy in BWTs, yet undetected by the RLE, and proposed a technique called \textit{tunnelling} purposed to deal with this particular redundancy. Tunnelled BWTs are a case of Wheeler graphs \cite{2017gagie_etal}, a general framework of the text indices based on the idea of suffix sorting.

Tunnelling has been experimentally proven to reduce the size of the resulting BWTs, especially in repetitive texts. The size of the input files is a significant bottleneck of this approach for two reasons. The first one is the fast BWT construction algorithm which requires linear, although non-negligible, memory overhead over the input file itself. The second reason is that tunnelling requires random access to the BWT, rendering it impractical when working in external memory.

The first of these issues could be alleviated by employing a suitable BWT construction algorithm, e.g. one that creates the BWT from a compressed representation of the input, such as the Lempel-Ziv parse by Policriti and Prezza \cite{2017policriti&prezza} or the Prefix-Free Parsing (PFP) method by Boucher et al. \cite{2019boucher_etal}.
PFP is especially interesting due to its simplicity and effectiveness, accomplishing to represent large repetitive texts in orders of magnitude smaller space \cite{2019boucher_etal}.
The second issue is harder to tackle, requiring the attention to either adapt the particular tunnelling algorithm to find the tunnelled BWT from the compressed representation or, when possible, tunnelling the BWT of the compressed representation and then using the information to output a tunnelled BWT of the original input.

The objective of our work is to resolve both problems simultaneously. We have adapted the PFP approach for building BWTs from large repetitive datasets in such a way that it produces tunnelled BWTs. While the resulting BWTs are larger than the original tunnelled BWTs, our approach scales to much larger dataset sizes and could be improved by further post-processing. Our method is not limited to producing tunnelled BWTs, but virtually any Wheeler graphs and hence provides a good starting point for building large pangenomic indexes.

\subsection{Preliminaries}
In this work, we will use the following notation.
Let $\Sigma$ be the working alphabet, i.e. a finite set of symbols totally ordered by $\preceq$.
We denote the size of $\Sigma$ as $|\Sigma| = \sigma$.
A \textit{string} $T = T[1] \ldots T[n]$ is a sequence of symbols from $\Sigma$.
We use the terms \textit{string} and \textit{text} interchangeably, similarly with \textit{symbol} and \textit{character}.
A string made of a single character is called \textit{unary}.
Furthermore, strings that end with $\texttt{\$}$ are called \textit{null-terminated}, where $\texttt{\$}$ is the lexicographically smallest symbol in $\Sigma$ (i.e. $\texttt{\$} \preceq c$ for any $c \in \Sigma$) that may only appear at the right end.
Moreover, we extend the use of $\preceq$ to label the lexicographic order of strings and write $\alpha \preceq \beta$ if and only if $\alpha$ is lexicographically smaller or equal than $\beta$, where $\alpha, \beta \in \Sigma^{*}$.
We will write $\alpha \prec \beta$ if $\alpha \preceq \beta$ and $\alpha \neq \beta$.

The length of the string $T$ will be denoted by $|T| = n$.
The string $\varepsilon$ is the unique string of length $0$.
A character of $T$ at position $i$ is denoted by $T[i]$.
A substring $T[i \ldots j]$ of $S$ is a subsequence $T[i] T[i+1] \ldots T[j]$.
We assume $T[i \ldots j] = \varepsilon$ if $i > j$.
The substring $T[1 \ldots j]$ is called the $j$-th \textit{prefix} of $T$, while the substring $T[j \ldots n]$ is called the $j$-th \textit{suffix}.

For integers $i, j$ we define the set $\{k ~|~ k \in \mathbb{N};~ i \leq k \leq j \}$ to be the \textit{interval} between $i$ and $j$ and denote it as $[i, j]$.

\section{Background}
\subsection{Burrows-Wheeler transform}
The Burrows-Wheeler transform (BWT) \cite{1994burrows&wheeler} is an incredibly influential transformation of text, formed by the concatenation of the letters in the last column of the Burrows-Wheeler matrix, which is created by lexicographically sorting all possible rotations of the original text.
This construction is visualized in Figure \ref{fig:bwt}.

\begin{figure}[ht]
\centering
\includegraphics[width=\linewidth]{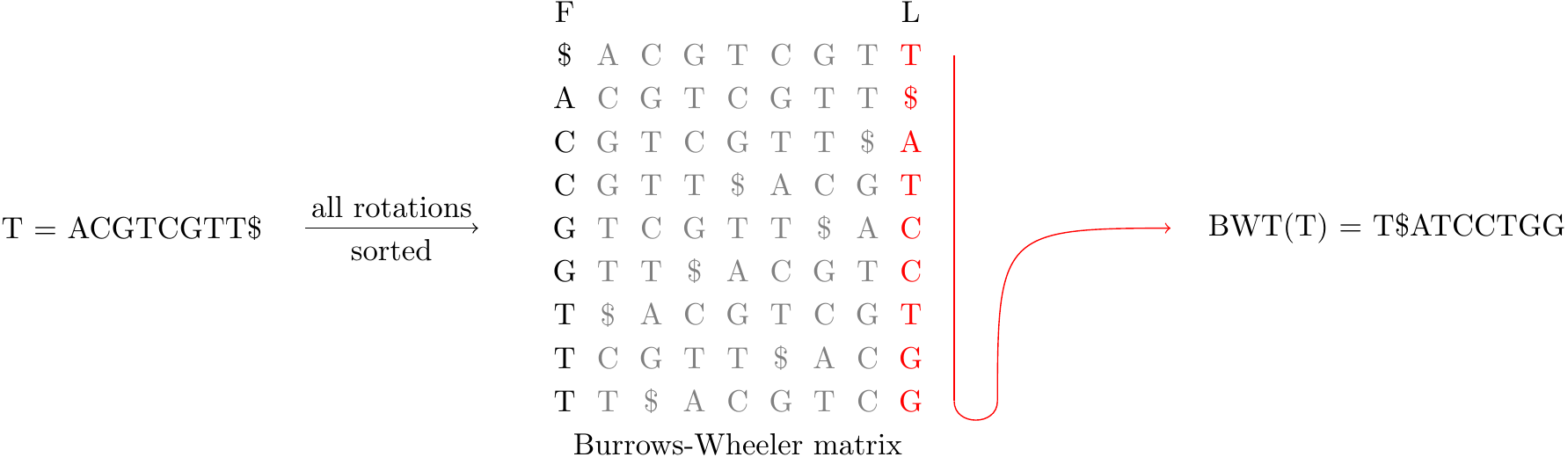}
\centering
\caption[BWT construction]{BWT construction. All rotations of the string \texttt{ACGTCGTT\$} are sorted lexicographically in the Burrows-Wheeler matrix. The last column of Burrows-Wheeler matrix represents the BWT.}
\label{fig:bwt}
\end{figure}

The importance of the BWT in string processing is derived from its compressibility, reversibility and usefulness in pattern matching.
At first, the improvement in the compressibility of BWT in comparison with the original string is not apparent. 
However, it stems from empirical observation that similar contexts tend to be preceded by the same characters.
The construction of BWT sorts these contexts and therefore tends to place identical characters in runs, which improves the compression.

The reversibility of the BWT is guaranteed by the so-called LF mapping, which uses the observation that the ranks of letters remain the same between the BWT and the first column of the Wheeler matrix F (i.e. the first letter in BWT corresponds to the first letter in F).
Since column F can be recovered from BWT by simply sorting the letters of BWT, the original text can be reconstructed in reverse order by following the LF mapping from the first row of the implicit Wheeler matrix.
Furthermore, the LF mapping, together with the suffix array, can be used for efficient pattern matching in the data structure known as FM-index \cite{2000ferragina&manzini}.

Although illustrative, the aforementioned construction of BWT is time-inefficient and seldom used in practice.
Several efficient methods for construction in linear time exist \cite{2003karkkainen&sanders, 2019boucher_etal}.
This work utilizes prefix-free parsing (PFP), a technique first used for building BWTs in \cite{2019boucher_etal}.
In comparison with previous methods, the PFP approach allowed the construction of BWT in sublinear space for repetitive data and therefore unlocked a possibility to produce BWTs of big datasets, of which sizes vastly exceed the memory limits of current computers.

\subsection{Prefix-free parsing}
In prefix-free parsing, the original text $T$ is divided into phrases of variable length, which are stored in the dictionary $\mathcal{D}$.
Furthermore, the parse $\mathcal{P}$ is created as a list of lexicographic ranks of phrases, in order in which the phrases appear in the original text.
The parse $\mathcal{P}$ and the dictionary $\mathcal{D}$ allow us to reconstruct the original text $T$.

To perform the division of $T$ into phrases, we define a set of trigger words $E$, where each trigger word is of length $w$.
Subsequently, we use a sliding window of size $w$ through the text $T$ and each time the sliding window matches a trigger word in $E$, we terminate the current phrase, add it to the dictionary, and initiate a new phrase.
In the end, the dictionary is sorted lexicographically, and the parse is relabeled accordingly.
It is noteworthy that the phrases in the dictionary begin and end with a trigger word, and consequently, no phrase is a prefix of another phrase in the dictionary.

This PFP construction is valuable for several reasons.
Firstly, only a single pass through the original text is needed to obtain all phrases, and therefore the text can be sequentially read from the disk allowing the processing of large texts.

Secondly, suppose the original text is repetitive. In that case, the parse and the dictionary tend to occupy a much smaller space than the raw representation, and therefore the PFP can be loaded into memory.

Finally, we can use a rolling hash $h: \text{word} \to \{0 \dots p-1\}$ to identify the trigger words.
If the hash is equal to zero, we consider it a member of the set of trigger words $E$.
The parameter $p$ then allows us to adjust the lengths of phrases since the expected length of a phrase in a uniformly-random text is $p$.

An example of a prefix-free parse construction is shown in Example \ref{ex:pfp}.

\begin{example}
\label{ex:pfp}
Let us have a text T = \texttt{\#ABDACDABDACDA\$}. For the purpose of this example, let the set of trigger words be $E = \{\texttt{A}\}$. Then we have $D = {\texttt{\#A}, \texttt{ABDA}, \texttt{ACDA}, \texttt{A\$}}$ and $\mathcal{P} = [0, 1, 2, 1, 2, 3]$. Then we lexicographically sort the dictionary phrases to get $\texttt{\#A}, \texttt{A\$}, \texttt{ABDA}, \texttt{ACDA}$ and the remapped parse $\mathcal{P} =	[0, 2, 3, 2, 3, 1]$.
\end{example}

\subsection{Wheeler graphs}
Wheeler graphs (WGs), introduced by Gagie et al. \cite{2017gagie_etal}, are a class of labelled graphs that generalize the notion of BWT and several of its variations into a unified framework.
A formal definition of Wheeler graphs is presented in Definition \ref{def:wheeler_graphs}. 

\begin{definition}(Wheeler graph)
    \label{def:wheeler_graphs}
    Let $G = (V, E)$ be a directed graph labeled by $\lambda: E \to \Sigma$. 
    Then $G$ is a Wheeler graph if and only if there is a total order $\leq$ of its vertices, also called the Wheeler order, such that for each pair of edges $e_1 = (u_1, v_1), e_2 = (u_2, v_2)$ the following conditions hold:
    \begin{enumerate}
        \item The vertices with zero in-degree precede those with non-zero in-degree,
        \item If $\lambda(e_1) \prec \lambda(e_2)$, then $v_1 < v_2$,
        \item If $\lambda(e_1) = \lambda(e_2)$ and $u_1 < u_2$, then $v_1 \leq v_2$.
    \end{enumerate}
\end{definition}

To illustrate this definition, we can visualize the Wheeler graphs by duplicating the vertices and displaying them in two columns in their Wheeler order, as shown in Figure \ref{fig:wg_viz}.
The edge $(u, v)$ will be drawn from vertex $u$ in the right column to vertex $v$ in the left column.
Then, the first condition guarantees that the vertices without an incoming edge will form a single interval in the ordering, and that interval will be placed at the top.
The second condition guarantees that the vertices with incoming edges labelled by a character $c$ will form an interval, and those intervals appear from top to bottom according to the lexicographic order of the character $c$.
The third condition guarantees that no two edges with the same label $c$ will cross.

\begin{figure}[!ht]
    \centering
    \includegraphics[width=0.55\linewidth]{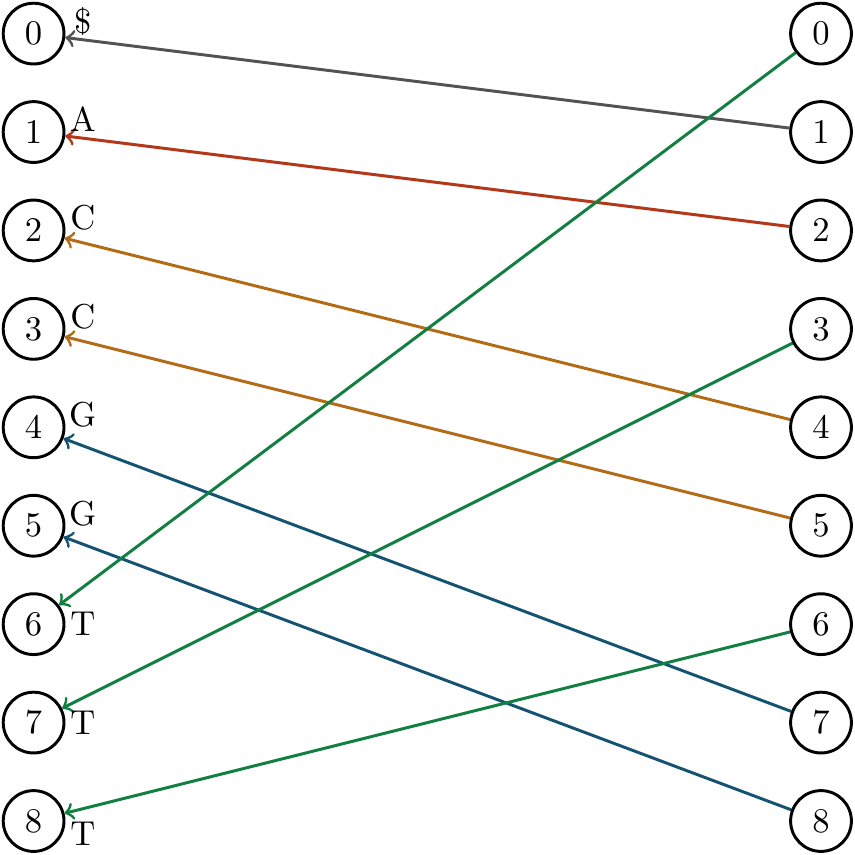}
    \caption{The Wheeler graph visualization. The vertices are doubled and shown in two columns to highlight the conditions of Definition \ref{def:wheeler_graphs}. Visualization also reveals the connection between WGs and LF mapping in BWTs.}
    \label{fig:wg_viz}
\end{figure}

\begin{figure}[!ht]
    \centering
    \includegraphics[width=0.85\linewidth]{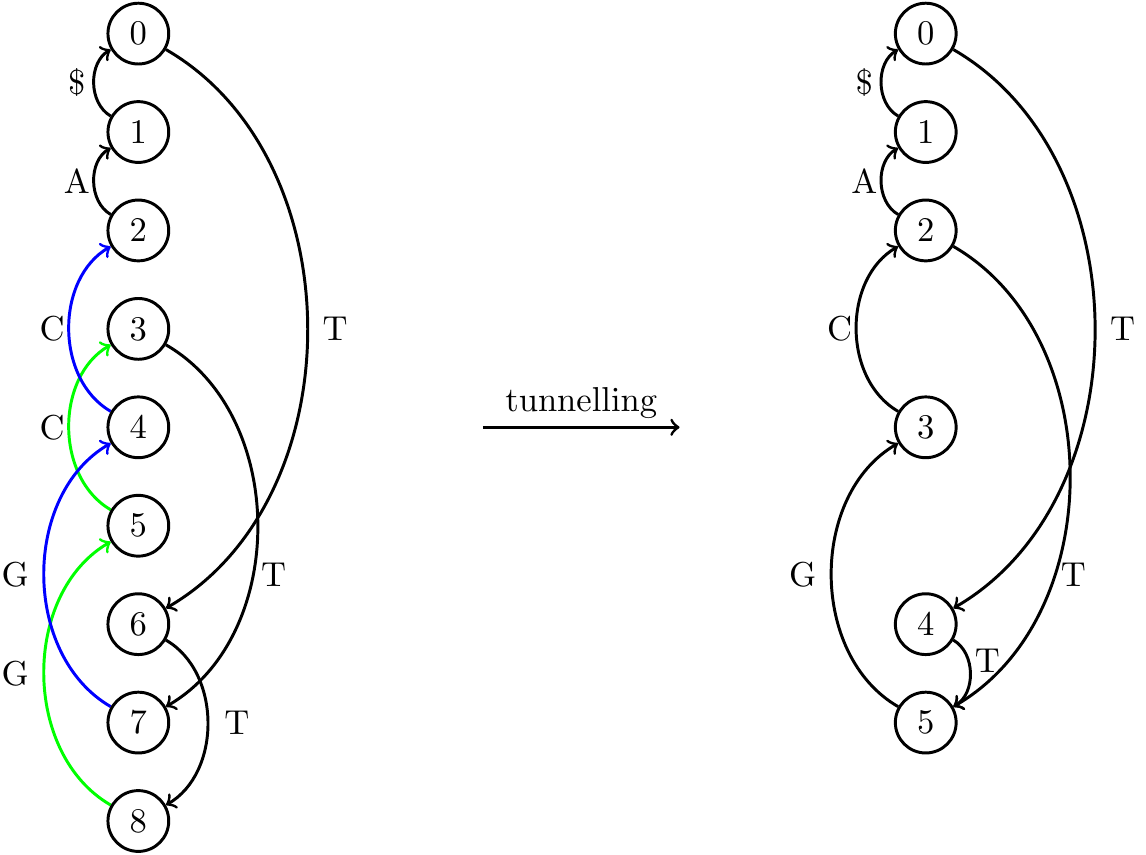}
    \caption{Tunnelling of the WG. Both paths $8-5-3$ (green) and $7-4-2$ (blue) spell \texttt{GC} and are merged into a single path in the reduced WG.}
    \label{fig:tunnelling}
\end{figure}

These conditions give rise to an essential property of WGs, called \textit{path coherence}, which is a generalization of the property from BWTs that enables the LF mapping.
Therefore, WGs can be used for efficient pattern matching and to store textual data in a lossless manner.
Furthermore, WGs allow for a succinct representation by storing the following data.
\begin{itemize}
    \item An array $C$ of length $\sigma$ s.t. $C[c]$ is the number of edges labeled by characters lexicographically smaller than $c$,
    \item The string $L = L_1 \ldots L_n$, where $L_i$ is the concatenation of the labels of the outgoing edges of vertex $v_i$,
    \item The bit vector $O = 10^{d_{out}(v_1)-1} \ldots 10^{d_{out}(v_n)-1}$,
    \item The bit vector $I = 10^{d_{in}(v_1)-1} \ldots 10^{d_{in}(v_n)-1}$,
\end{itemize}
where $d_{in}(v)$ ($d_{out}(v)$) is the in-degree (out-degree) of the vertex $v$.

In comparison with the plain list-of-edges representation, which uses $\Theta (|E| \log \sigma + |E| \log |V|)$ bits of space, the succinct representation uses only $\Theta (|E| \log \sigma + 2|E| + \sigma \log |E|)$ and together with efficient rank and select queries on the bit vectors $O$ and $I$ supports efficient graph traversing. 

Another compelling property of WGs is their reducibility, where a WG representing a particular set of strings can be replaced by a smaller WG representing the same set of strings.
One way to accomplish this is by treating the given WG as a Wheeler nondeterministic finite automaton (WNFA), which we subsequently convert to its deterministic variant (WDFA).
It has been shown that such conversions for automata that correspond to WGs result only in a linear number of additional states \cite{2020alanko_etal}.
For a WDFA, an equivalent WDFA with the minimum number of states can then be found by a linear-time algorithm by Alanko et al. \cite{2022alanko_etal}.

Another particular process of reducing a WG is called \textit{tunnelling} \cite{2018baier} and is illustrated in Figure \ref{fig:tunnelling}.
During the tunnelling, the paths $p_1 = (v_{1,1}, \ldots,v_{1, \ell}), \ldots, p_k = (v_{k,1}, \ldots,v_{k,\ell})$ of which the vertex sets $V_j = \{v_{i,j} ~|~ 1 \leq i \leq k\}$ are adjacent in the Wheeler order for each $1 \leq j \leq \ell$, and which spell the same substring are identified and a tunnel is formed by merging the vertices from $V_j$ into a single vertex $v_j$.
In general, the task of finding the smallest WG via tunneling was proven to be NP-hard, and therefore multiple heuristic approaches to solve this problem were proposed \cite{2019baier&dede}.

\section{Methods}

\label{methods}

We propose to employ PFP as a preprocessing step in building Wheeler graph indexes from large datasets through tunnelling BWTs.
We are inspired by the original work by Boucher et al. \cite{2019boucher_etal} of building large BWTs using small working memory.

After obtaining the parse $\mathcal{P}$ of a text $T$, we build the $BWT(\mathcal{P})$ and subsequently find its tunnelling using one of the approaches proposed by Baier and Dede \cite{2019baier&dede}. 
Since the parse $\mathcal{P}$ is a sequence of lexicographic ranks of dictionary phrases, each character of the parse $\mathcal{P}$ corresponds to multiple characters of $T$ (controlled to a certain degree by the PFP parameter $p$), and hence serves as a significantly shortened representation of $T$.
Therefore, tunnelling $BWT(\mathcal{P})$ is practical for much larger file sizes.

Under the additional assumption that the input files are repetitive, we can also deduce that at least a certain degree of repetitiveness will also be present in $BWT(\mathcal{P})$, and hence $BWT(\mathcal{P})$ can be significantly shortened by tunnelling. Once we obtain a tunnelled $BWT(\mathcal{P})$ in the form of a Wheeler graph $G_P$ of the parse $\mathcal{P}$, our goal is to efficiently transform it to the Wheeler graph $G_T$ of the text $T$ while retaining the tunnelled paths in an appropriate form.


\subsection{Theory}
We begin by introducing a vital claim that was already proven by Boucher et al. \cite{2019boucher_etal}.

\begin{definition}[Suffix set \cite{2019boucher_etal}]
\label{def:suffixset}
Let $T$ be a text of length $n$. We call the set $\text{Suf}(T) = \{T[i \ldots n] ~|~ i \in [1, n]\}$ a suffix set of $T$.
\end{definition}

\begin{lemma}
\label{lemma:suffix_comparison}
Let $T$ be a text, $\mathcal{P}$ be the parse of the $PFP(T)$ and $x, y \in \text{Suf}(T)$ such that $x \prec y$. Furthermore, let $\alpha_x, \alpha_y$ be prefixes of $x,y$ up to the nearest phrase boundaries and $\beta_x, \beta_y$ be the rest of $x,y$.
The $\beta_x, \beta_y$ correspond to suffixes $\gamma_x, \gamma_y$ in the parse $\mathcal{P}$. 
Then either it holds that $\alpha_x \prec \alpha_y$ or $\alpha_x = \alpha_y$ and $\gamma_x \prec \gamma_y$.
\end{lemma}

\begin{proof}
The claim follows from the lemmas 1-7 proven by Boucher et al. \cite{2019boucher_etal}.
\end{proof}

The lemma \ref{lemma:suffix_comparison} allows the $BWT(T)$ to be constructed from the parse $\mathcal{P}$ and the dictionary $\mathcal{D}$ by iterating over the suffixes of the phrases from $\mathcal{D} = \{ D_1, \ldots, D_k \}$ sorted in the lexicographic order and outputting the characters preceding their occurrences in $T$.
As the vertices of Wheeler graphs are totally ordered according to the suffixes of $T$ they represent, we immediately obtain a straightforward generalization of this result for Wheeler graphs in the sense that a Wheeler graph of $\mathcal{P}$ implies a Wheeler graph of $T$. The exact meaning of this implication is detailed in Def. \ref{def:expandedWG}.

\begin{definition}[Expanded graph]
\label{def:expandedWG}
Let $T$ be a string and $G_{\mathcal{P}} = (V_{\mathcal{P}} = [1, m], E_{\mathcal{P}}, \lambda_{\mathcal{P}})$ be a Wheeler graph of $\mathcal{P}$ from the $PFP(T)$, the vertices of which are labeled according to a Wheeler order. Then we define the graph $G' = (V', E', \lambda')$ to be the Expanded graph of $T$ constructed from $G_{\mathcal{P}}$, where

\begin{align*}
V_{\text{int}} = & \{v_{e,2},\ldots,v_{e,\ell-w} ~|~ e \in E_{\mathcal{P}}, \ell = |D_{\lambda_{\mathcal{P}}(e)}|\},\\
V' = & V_{\mathcal{P}} \cup V_{\textrm{int}},\\
E' = &\{(v_{e,i}, v_{e,i-1}) ~|~ e \in E_{\mathcal{P}}, i \in [3, \ell-w], \ell = |D_{\lambda_{\mathcal{P}}(e)}|\} ~\cup\\
 &\{(t, v_{e, \ell-w}) ~|~ e = (t, u) \in E_{\mathcal{P}}, \ell = |D_{\lambda_{\mathcal{P}}(e)}|\} ~\cup \\
 &\{ (v_{e, 2}, u) ~|~ e = (t, u) \in E_{\mathcal{P}} \} \\
\lambda' (e') = &
\begin{cases}
D_{\lambda_{\mathcal{P}}(e)}[i-1] & \text{if } e' =
(v_{e, i}, v_{e, i-1}), e \in E_{\mathcal{P}},
\ell = |D_{\lambda_{\mathcal{P}}(e)}|, i \in [3, \ell-w],\\
D_{\lambda_{\mathcal{P}}(e)}[\ell-w] & \text{if } e' = (t, v_{e, \ell-w}),  e \in E_{\mathcal{P}},  \ell = |D_{\lambda_{\mathcal{P}}(e)}|,\\
D_{\lambda_{\mathcal{P}}(e)}[1] & \text{if } e' = (v_{e, 2}, u), e = (v, u) \in E_{\mathcal{P}}.\\
\end{cases}\\
\end{align*}
\end{definition}

The expanded graph $G_T = (V, E, \lambda)$ is therefore obtained by expanding each edge $e = (t, u) \in E_{\mathcal{P}}$ labeled $\lambda_{\mathcal{P}}(e)$ into a path with newly added internal vertices $v_{e,2}, \ldots, v_{e,\ell-w}$ so that the concatenation of the labels $\lambda(( v_{e, 2}, u)) \lambda((v_{e,3}, v_{e,2})) \ldots \lambda((v_{e, \ell-w}, v_{e,\ell-w-1})) \lambda((t, v_{e,\ell-w})) = D_{\lambda(e)}[1 \ldots \ell - w]$, where $\ell = |D_{\lambda(e)}|$.
In other words, if we take an edge $e = (t, u) \in E_{\mathcal{P}}$ and an internal vertex $v$ from the expansion of $e$ in $G_T$, then starting from $v$ and outputting the edge labels while traversing the reversed edges up to $t$ we get a suffix of $D_{\lambda(e)}[1 \ldots \ell-w]$. With the help of Def. \ref{def:incoming}, we extend this notion to non-internal vertices in Def. \ref{def:phrase_suff}, which we then use in Thm. \ref{thm:core}.

\begin{definition}[Incoming label]
\label{def:incoming}
Let $G = (V, E, \lambda)$ be a Wheeler graph. Then for any $v \in V$ we will call the value $i(v) = \lambda((u, v))$ for such $u \in V$ that $(u, v) \in E$ an incoming label of the vertex $v$.
\end{definition}

\begin{definition}[Phrase suffix]
\label{def:phrase_suff}
Let $G_{\mathcal{P}} = (V_{\mathcal{P}}, E_{\mathcal{P}}, \lambda_{\mathcal{P}})$ be a
Wheeler graph over a PFP of the string $T$. Let $G = (V, E, \lambda)$ be the expanded graph of $G_{\mathcal{P}}$. Then for any $v \in V$ we defined $f(v)$ as

$$
f(v) = \begin{cases}
			D_{i(w)}[i \ldots \ell-w], & \text{if $v = v_{e,i}, e = (u, w) \in E_{\mathcal{P}}, i \in [2,\ell-w], \ell = |D_{\lambda_{\mathcal{P}}(e)}|$}\\
            D_{i(v)}[1 \ldots \ell-w], & \text{otherwise}
		 \end{cases}
$$
\end{definition}

Note that in a Wheeler graph, all the incoming edges to a vertex $v$ have to be labelled the same, hence the correctness of Def. \ref{def:incoming}. Def. \ref{def:phrase_suff} deals with two cases -- the first one being the internal vertex of the expanded Wheeler graph and a vertex of the $G_{\mathcal{P}}$. Now we are ready to state our core theorem.

\begin{theorem}
\label{thm:core}
Let $G_{\mathcal{P}} = (V_{\mathcal{P}}, E_{\mathcal{P}})$ be a
Wheeler graph over a PFP of the string $T$. Let $G = (V, E)$ be the expanded graph of $G_{\mathcal{P}}$. Then $G$ is also a Wheeler graph.
\end{theorem}

\begin{proof}
We prove the claim by showing that there is an ordering of vertices from $V$ such that the conditions from Def. \ref{def:wheeler_graphs} are fulfilled.
We define the order as follows.
Let $v_1, v_2 \in V$ and $f(v_i)$ a phrase suffix according to Def. \ref{def:phrase_suff} for $i = 1,2$. Then $v_1 \prec v_2$ if and only if $f(v_1) \prec f(v_2)$ or $f(v_1) = f(v_2)$ and $u_1 < u_2$.

We have no vertices of zero in-degree in $G$, so property 1 is satisfied for any vertex order. We need to show that the vertices in $V$ are sorted according to the suffixes of $T$ they represent. However, this follows from the lemma \ref{lemma:suffix_comparison} and the definition of our order.
\end{proof}

We note that finding an ordering of the nodes of Wheeler graphs is \textit{NP}-complete in general \cite{2022gibney&thankachan}; however, it can indeed be accomplished in a polynomial time in our restricted case, in which we already have the relative order of the vertices with in-degree and out-degree larger than one\footnote{we even have the relative order of some vertices with both in-degree and out-degree equal to $1$, but that is not helpful in general}.

To build the expanded Wheeler graph, we follow an approach similar to that of Boucher et al. \cite{2019boucher_etal} for building large BWTs from $\mathcal{P}$ and $\mathcal{D}$. First, we construct a suffix array SA$_\mathcal{D}$ of the concatenation of all phrases from the dictionary $\mathcal{D}$ in the lexicographic order. As we pointed out before, to find the relative order of the nodes $u, v$ in a Wheeler graph, we need to compare the suffixes of $T$ that precede $u$ and $v$.
According to the proof of Theorem \ref{thm:core}, it only suffices to compare the suffixes $\alpha, \beta$ of the phrases in which $u, v$ are located, and if they are equal, we need only to compare their following suffixes, which we can do using the parse $\mathcal{P}$. Having already computed the Wheeler graph $G_\mathcal{P}$, we know the lexicographic rank of each suffix among those starting at phrase boundaries.

Hence, as in Boucher et al. \cite{2019boucher_etal}, we will iterate over the SA$_{\mathcal{D}}$ while we simultaneously build a succinct representation of $G_T$, i.e. the array $L$ together with bit vectors $I$ and $O$. For each phrase $d \in \mathcal{D}$, we save the ranks of vertices $u \in V$ such that $e = (u,v) \in E$ and $\lambda(e) = d$. Moreover, for each such phrase suffix $s$, we save the labels of the outgoing edges from the vertices $v \in V$ that are suffixed by $s$.

We will proceed in two passes, each time outputting the entries of the arrays $L$, $I$ and $O$ for a different subset of nodes. First, we will process those that correspond to vertices suffixed by those elements $s \in \mathcal{S}$ that are either suffix of only a single phrase $d \in \mathcal{D}$ or appear only once in $\mathcal{D}$, leaving placeholders for the rest of the vertices. The vertices corresponding to elements of $\mathcal{S}$ that are suffixes of multiple phrases of $\mathcal{D}$ or are a whole phrase that appears multiple times in $T$ will be processed according to the increasing Wheeler order of their phrase-boundary suffixes. All in all, this proves Theorem \ref{thm:result}.

\begin{theorem}
\label{thm:result}
Let $G_{\mathcal{P}} = (V_{\mathcal{P}}, E_{\mathcal{P}}, \lambda_{\mathcal{P}})$ be a
Wheeler graph over a PFP of the string $T$ of length $n$. Then we can build the expanded Wheeler graph $G = (V, E, \lambda)$ from $G_{\mathcal{P}}$ in $O(n)$ time and workspace proportional to $O(PFP(T))$.
\end{theorem}

\subsection{Implementation}
We have built our approach as an extension of the \texttt{bigbwt} -- an implementation of the original method by Boucher et al. \cite{2019boucher_etal} for building BWTs using PFP.
We have used the original PFP component of the \texttt{bigbwt} without changes. In contrast with \texttt{bigbwt}, the only intermediate file that we need apart from the PFP is the tunnelled BWT of the input.
We have used the linear time SACA-K \cite{2013nong} algorithm that only uses $O(1)$ workspace to create the suffix array of $\mathcal{P}$, which is subsequently transformed into the $BWT(\mathcal{P})$.
As in \texttt{bigbwt}, the suffix array of $\mathcal{D}$ along with the LCP array is created using the gSACAK algorithm \cite{2017louza_etal} and is later used to iterate over the suffixes of the phrases from $\mathcal{D}$ in the increasing lexicographic order.

To find the tunnelling of the $BWT(\mathcal{P})$, we considered several approaches suggested by Baier et al. \cite{2018baier, 2019baier&dede, 2020baier_etal}. Most of them aimed at reducing the compressed size of the run-length encoded BWT.
Since our current goal is not to optimize the compressed size of the constructed WG (which is not straightforward to do when only having access to $\mathcal{P}$), we have employed a tunnelling method which does not necessarily minimize the size of the tunnelled BWT after compression, but rather the number of edges in the resulting WG.
While being tightly linked with the edge-minimization of the de Bruijn graph constructed from the input text, the prefix intervals in this method do not overlap, and Baier et al. \cite{2020baier_etal} have shown that the method leads to significant reduction of BWT size for repetitive datasets.
The only alteration we have made is that we have enabled the processing of large integer alphabets.

The implementation is heavily based on the Succinct Data Structure Library (SDSL) by Simon Gog et al.\cite{2014gog_etal}.
Namely, the tunnelled BWT is represented as a wavelet tree that stores the $L$ component, which allows for efficient select queries.
The rank and select support for the bitvectors $I$ and $O$ is also supported by SDSL.

Similarly to the implementation of the \texttt{bigbwt} \footnote{\url{https://gitlab.com/manzai/Big-BWT}} \cite{2019boucher_etal}, we slightly deviate from the method of resolving the ambiguous phrase suffixes; instead of leaving the placeholders to fill in another pass, we use heapsort to merge the given vertices according to the rest of the suffixes. We note that this information is already available in the $G_{\mathcal{P}}$ as the labels of the vertices in $V_{\mathcal{P}}$, so we only need to sort integers instead of actual suffixes. This approach allows us to sequentially write the succinct representation of the resulting WG in a single pass, to some extent alleviating the slow writing process of external memories.
\section{Experiments}
We have demonstrated the applicability of our approach using two real-world datasets, 5520 genomes of the Salmonella genus and 1000 copies of human chromosome 19.

The Salmonella dataset was obtained using \texttt{ncbi-datasets-cli} tool.
Particularly, the command \texttt{datasets download genome taxon salmonella --assembly-level complete\_genome} was used.
The genomic fasta files were concatenated into a single fasta file of size 13GB. The dataset of haplotypes of human chromosome 19 was obtained from the 1000 genomes project \cite{2008siva}\footnote{dataset available at \url{http://dolomit.cs.tu-dortmund.de/tudocomp/}}.

From these fasta files, we extracted subsets of increasing sizes to observe the trends of running time and memory consumption, as well as the sizes of the resulting WGs. For the sake of simplicity, we preprocessed the datasets so that we ignored all characters apart from \texttt{A}, \texttt{C}, \texttt{G}, \texttt{T}.
All of the experiments were run on a machine with Intel(R) Xeon(R) CPU E5-2670 0 @ 2.60GHz and 144GB of operational memory.
Every program was executed using a single thread. The running time and peak memory consumption were measured using the \texttt{/usr/bin/time -v} command.

The results of the experiments for the copies of the Salmonella genus and chromosome 19 are displayed in Table \ref{tab:salm} and \ref{tab:chr19}.
Apart from our implementation, we also provide the results for the approach of Baier et al. \cite{2020baier_etal} for comparison.
Their approach consists of first constructing a BWT of the input and then tunnelling the constructed BWT.
The BWT construction algorithm can be done by either the DIVSUFSORT algorithm \cite{2017fischer&kurpicz} or a semi-external variant of the induced sorting algorithm \cite{2009nong_etal} (SE-SAIS), where the former is fast but demanding significantly more memory than the latter, which is slower but memory-efficient.
We have used the SE-SAIS since it makes a more fair comparison to our PFP-based algorithm.
Our experiments were all ran with the PFP parameters set to $w = 4$ and $p = 50$. 

Expectedly, the WGs we produce are larger than those of Baier et al. due to the constraints of PFP, which limits the space of possibilities of tunnelling and treats the removal of any parse character such as having the same benefit, which is clearly not the case from the perspective of the input.
We also point out that smaller WGs can be achieved with different combinations of the PFP parameters.
Our approach, however, operates within much less memory and shorter running time, which allows it to construct WGs for datasets of unprecedented sizes, e.g. the whole 1000 human genomes project and beyond.

\begin{table}[ht]
\centering
\begin{tabular}{|ll|rrr|rrr|} \hline
\multirow{2}{*}{\#sequences} & \multirow{2}{*}{size} & \multicolumn{3}{c|}{Baier et al. \cite{2020baier_etal}} & \multicolumn{3}{c|}{PFP} \\ \cline{3-8}
& & time & memory & size & time & memory & size\\ \hline
1 & 5 & 8 & 47 & 6 & 5 & 57 & 7\\
10 & 34 & 42 & 118 & 16 & 24 & 223 & 38\\
100 & 238 & 301 & 389 & 55 & 99 & 696 & 174\\
1000 & 2619 & 3431 & 3128 & 430 & 568 & 2709 & 1400\\
5000 & 12053 & 18658 & 14538 & 1275 & 1726 & 4859 & 4189\\ \hline
\end{tabular}
\caption{Results of the experiments on the Salmonella genomes. The 'time' is the running time in seconds, 'memory' stands for the peak memory usage in MBs during the running time. The file sizes are reported in MBs.}
\label{tab:salm}
\end{table}

\begin{table}[ht]
\centering
\begin{tabular}{|ll|rrr|rrr|} \hline
\multirow{2}{*}{\#sequences} & \multirow{2}{*}{size} & \multicolumn{3}{c|}{Baier et al. \cite{2020baier_etal}} & \multicolumn{3}{c|}{PFP} \\ \cline{3-8}
     & & time & memory & size & time & memory & size\\ \hline
    1    &  54      &   81      &   341     & 67 & 82 & 597 & 70\\
    10   &  533     &   700     &   708     & 87 & 157 & 880 & 296\\
    100  &  5322    &   7263    &   6748    & 162 & 644 & 1319 & 2123\\
    500  &  27903   &  54498    &   34031    & 749 & 2634 & 4190 & 13110\\
    1000 &  53220   & --- & --- & --- & 5938 & 8009 & 29803\\ \hline
\end{tabular}
\caption{Results of the experiments on the copies of chromosome 19. The 'time' is the running time in seconds, 'memory' stands for the peak memory usage in MBs during the running time. The file sizes are reported in MBs. The Baier et al.'s \cite{2020baier_etal} approach for 1000 sequences was terminated after 20 hours.}
\label{tab:chr19}
\end{table}

\section{Conclusion and future work}
We have successfully demonstrated that the prefix-free parsing technique can be used to alleviate the computational requirements of the construction of tunnelled BWTs.
Our approach allows the use of Wheeler graphs as pangenomic references for huge datasets such as the 1000 Genomes Project, the Vertebrate Genomes Project, the Earth Microbiome Project, and many more. 

To this end, we generalized the approach of Boucher et al. \cite{2019boucher_etal} devised for the construction of BWTs from large volumes of repetitive data and experimentally showed it can be a good starting point for building Wheeler graphs.

Since our approach treats the PFP parse as any input text, it does not exploit the full information the PFP provides and could benefit from incorporating the lengths of the dictionary phrases into the tunnelling process, potentially even to the point of simulating the tunnelling algorithm on the original text and producing the same output. We leave this for future work.

Another logical next step is to enhance our approach with the ability to output a sufficiently small suffix array sample to allow our Wheeler graphs to locate the occurrences of patterns.

\bibliography{content/main}

\appendix

\end{document}